\documentclass[runningheads,orivec]{llncs}

\usepackage[T1]{fontenc}
\usepackage{graphicx}
\usepackage{amsmath}

\DeclareMathOperator*{\argmin}{arg\,min}
\usepackage{ulem}
\usepackage{float}
\usepackage{subcaption}
\usepackage[ruled,vlined]{algorithm2e}
\newcommand\algoName{i-QLS }
\usepackage{amsfonts}
\usepackage{amsmath}
\usepackage{hyperref}
\usepackage{color}

\newcommand\blfootnote[1]{%
  \begingroup
  \renewcommand\thefootnote{}\footnote{#1}%
  \addtocounter{footnote}{-1}%
  \endgroup
}

\begin{document}

\title{i-QLS: Quantum-supported Algorithm for Least Squares Optimization in Non-Linear Regression
}

\titlerunning{i-QLS: Quantum-supported Algorithm for Least Squares}

\author{Supreeth Mysore Venkatesh\inst{1,2} \and
Antonio Macaluso\inst{2} \and
Diego Arenas\inst{2} \and
Matthias Klusch\inst{2} \and
Andreas Dengel\inst{1,2}}

\authorrunning{Venkatesh et al.}

\institute{University of Kaiserslautern-Landau (RPTU), Kaiserslautern, Germany \email{\{supreeth.mysore\}@rptu.de} \and
German Research Center for Artificial Intelligence (DFKI), Germany
\email{\{supreeth.mysore, antonio.macaluso, diego.arenas matthias.klusch, andreas.dengel\}@dfki.de}}

\maketitle
\begin{abstract}
We propose an iterative quantum-assisted least squares (i-QLS) optimization method that leverages quantum annealing to overcome the scalability and precision limitations of prior quantum least squares approaches. Unlike traditional QUBO-based formulations, which suffer from a qubit overhead due to fixed discretization, our approach refines the solution space iteratively, enabling exponential convergence while maintaining a constant qubit requirement per iteration. This iterative refinement transforms the problem into an anytime algorithm, allowing for flexible computational trade-offs. Furthermore, we extend our framework beyond linear regression to non-linear function approximation via spline-based modeling, demonstrating its adaptability to complex regression tasks. We empirically validate i-QLS on the D-Wave quantum annealer, showing that our method efficiently scales to high-dimensional problems, achieving competitive accuracy with classical solvers while outperforming prior quantum approaches. Experiments confirm that i-QLS enables near-term quantum hardware to perform regression tasks with improved precision and scalability, paving the way for practical quantum-assisted machine learning applications.
\blfootnote{This work has been accepted for publication in the proceedings of the 25th International Conference on Computational Science (ICCS 2025), held 7--9 July 2025, Singapore.}

\end{abstract}

\section{Background and Related Works}





A key distinction in machine learning methodologies lies in their optimization strategies, which are shaped by the assumptions that models make about the relationship between input features and the target variable. Models as Neural networks rely on non-convex optimization, often requiring extensive datasets and prolonged training due to the difficulty of escaping local minima \cite{goodfellow2016deep}. In contrast, parametric models based on least squares (LS) optimization \cite{hastie01statisticallearning}, such as splines \cite{reinsch1967smoothing} and support vector machines \cite{cortes1995support}, benefit from convex optimization, ensuring global optimality and theoretical robustness. However, the polynomial complexity of matrix operations in LS optimization introduces significant computational constraints as the number of features grows, limiting scalability.

Recently, several quantum machine learning models have been proposed for supervised learning tasks\cite{macaluso2024quantum}. However, relatively few studies have explored the potential of leveraging near-term quantum computing exclusively for training and parameter estimation while maintaining a classical model for inference \cite{sinha2025nav,jerbi2024shadows}. 
In particular, quantum annealing has been proposed as a method for reformulating LS optimization into a Quadratic Unconstrained Binary Optimization (QUBO) problem for execution on quantum hardware \cite{Cruz-Santos2019,10.1007/978-3-030-10564-8_23,Date2021}. 
Despite these advancements, quantum annealing approaches face two fundamental bottlenecks. Specifically, scalability remains a major challenge, as a tradeoff exists between the number of qubits and the precision of estimates of the optimized weights.
Also, these methods are primarily to linear models, which fail to capture the complexities of real-world problems.
Quantum annealing has been recently explored as an alternative paradigm for solving combinatorial optimization problems that naturally map onto a QUBO formulation. Although the LS problem in regression being continuous optimization does not inherently conform to the QUBO framework which is combinatorial optimization, few studies have adapted the formulation to harness quantum annealing for regression purposes.

For instance, \cite{Date2021} and \cite{10.1007/978-3-030-10564-8_23} discretize the solution space by encoding each continuous weight as a series of binary variables, thereby formulating a corresponding QUBO problem that is amenable to solution via a quantum annealer. A significant limitation of these approaches is that increasing the precision of the weights necessitates an exponential increase in the number of qubits, rendering the method infeasible for large-scale problems on current quantum hardware. Consequently, classical LS solvers—such as direct matrix inversion \cite{woodbury1950inverting}, QR decomposition  \cite{10.1093/comjnl/4.3.265}, or iterative gradient-based methods \cite{lanczos1952solution}—often outperform quantum approaches due to their capacity to achieve high precision without incurring the substantial overhead associated with quantum hardware.

Moreover, existing quantum annealing methods for LS have predominantly focused on linear regression, thereby neglecting non-linear problems that are more representative of real-world applications \cite{koura2024linearregressionusingquantum}. In contrast, gate-based quantum computation has addressed non-linear approximation through the development of quantum splines \cite{macaluso2020quantum,inajetovic2023enabling}. Splines, which are piecewise polynomial functions with smoothness constraints, provide an effective tool for modeling non-linear relationships within a structured regression framework. In the initial formulation of quantum splines \cite{macaluso2020quantum}, a specific LS formulation based on B-splines \cite{Boor1971} was employed to exploit the computational advantages of the HHL algorithm \cite{harrow2009quantum} for solving sparse linear systems. Although this approach offers a theoretical computational advantage, its reliance on fault-tolerant quantum computation limits its near-term applicability. Consequently, a variational counterpart has been proposed to leverage near-term quantum devices \cite{inajetovic2023enabling}, although it has not yet demonstrated a robust advantage over classical methods.

Our work builds on these methodologies but fundamentally differs in two key aspects. First, rather than formulating a fixed QUBO problem with a predetermined discretization, we introduce an iterative refinement process. In our approach, each QUBO solution serves as the input for the subsequent iteration, progressively narrowing the weight search space while maintaining a fixed qubit requirement. This iterative refinement accelerates convergence exponentially and overcomes the scalability limitations encountered in previous methods. Our method is inspired by classical iterative refinement techniques commonly used in numerical optimization—such as Newton’s method\cite{more1982newton}, gradient descent\cite{haji2021comparison}, and expectation-maximization \cite{moon1996expectation}. 
To the best of our knowledge, apart from combinatorial optimization problem like graph-clustering \cite{gcsq,quacs}, previous QUBO-based quantum regression methods have not incorporated an iterative refinement mechanism, rendering our approach unique in this context.

Second, our approach explicitly extends quantum-assisted regression to non-linear function approximation by integrating quantum spline formulations. In our framework, the advantages of splines—namely, their ability to model non-linear relationships through piecewise polynomial fitting with smoothness constraints—are harnessed within the iterative QUBO process to determine optimal spline coefficients. This integration enables the capture of non-linear dynamics while mitigating the fixed discretization issues of prior methods. By unifying iterative refinement with spline-based modeling, our method addresses the limitations of both conventional quantum annealing approaches and the earlier quantum spline formulations that rely on fault-tolerant computation.

\section{Methods}\label{sec:methods}

In this section, we present the mathematical formulation and derivation of our proposed i-QLS algorithm. We begin by posing the standard linear regression model as a discretized optimization problem, then describe how to embed it into a QUBO form suitable for a quantum annealer. Subsequently, we extend the iterative scheme to refine the precision of the solution at each iteration without requiring a large number of qubits from the outset.

\vspace{-8pt}

\subsection{Problem Formulation}
\vspace{-6pt}
Let 
\[
\mathbf{X} \in \mathbb{R}^{N \times d}, \quad 
\mathbf{y} \in \mathbb{R}^N,
\]
where \(N\) is the number of data points and \(d\) is the number of features (or variables). The goal of linear regression is to find a weight vector 
\[
\mathbf{w^*} = (w_1^*, w_2^*, \dots, w_d^*)^\top
\]
that minimizes the sum of squared errors (SSE):
\begin{equation}
\mathbf{w^*} = \argmin_{\mathbf{w} \in \mathbb{R}^d} \; S(\mathbf{w}) 
\;=\; 
\sum_{n=1}^N \bigl(y_n - \mathbf{w}^\top \mathbf{x}_n\bigr)^2,
\end{equation}
where \(\mathbf{x}_n\) is the \(n\)-th sample. 
The least squares solution can be found in closed form via \(\mathbf{w} = (\mathbf{X}^\top \mathbf{X})^{-1}\mathbf{X}^\top \mathbf{y}\) in the non-singular case. However, our method operates on a discretized search space to allow a quantum annealer to sample candidate \(\mathbf{w}\) values.

\subsection{Discretization and QUBO Construction}

Representing each weight parameter \(w_i \in \mathbb{R}\) by $m$ binary variables and initializing a sufficiently large interval $\Delta_i^{(0)}$ with lower bound $\ell_i^{(0)}$ and the upper bound $u_i^{(0)}$, we find the discretization step size:
\begin{equation}
\begin{split}
\delta_i^{(0)} \;=\; \frac{\Delta_i^{(0)}}{2^m-1} \quad \text{for} \quad i=1,\dots,d, \qquad
\text{where} \qquad \Delta_i^{(0)} \;=\; (u_i^{(0)} - \ell_i^{(0)})
\end{split}
\end{equation}

Initializing $w_i^{(0)} \gets \frac{(u_i^{(0)} + \ell_i^{(0)})}{2}$, our goal is to find $w_i^{(1)}$ restricted to $2^m$ equally spaced values in the interval $[\ell_i^{(0)},u_i^{(0)}]$ that is closest to the optimal weight $w_i^*$. Mathematically,

\vspace{-10pt}

\begin{equation}
\begin{split}
w_i^{(1)*} = \argmin_{w_i^{(1)}} |w_i^{(1)} - w_i^*|, \\
\text{where} \quad w_i^{(1)} \in \{\ell_i^{(0)}+\delta_i^{(0)}p~|~p \in \{0,1,\dots,2^{m}-1\}\}
\end{split}
\end{equation}
\begin{equation}
\mathbf{b} \;=\; 
\bigl(b_{1,0}, b_{1,1},\dots, b_{1,m-1}, \dots, b_{d,0}, \dots, b_{d,m-1}\bigr)^\top
\end{equation}
be the binary encoding vector representing the discretized weights. The weight vector \(\mathbf{w}\) can now be expressed as a linear function of \(\mathbf{b}\):
\begin{equation} \label{eqn:weight_vector_discretization}
\begin{split}
\mathbf{w^{(1)}}(\mathbf{b}) 
\;=\; 
\boldsymbol{\ell^{(0)}} 
\;+\; \mathbf{D^{(0)}} \mathbf{B}(\mathbf{b}),
\end{split}
\end{equation}
$$
\text{where} \quad \boldsymbol{\ell} = (\ell_1, \dots, \ell_d)^\top, \\
$$
$\mathbf{D}$ is a diagonal matrix containing the discretization step sizes $\delta_i$,
$$
\mathbf{D} = \text{diag}(\delta_1, \dots, \delta_d), 
$$
and $\mathbf{B}(\mathbf{b})$ is the vector whose entries are sums of powers of two weighted by the bits i.e.,
\begin{equation}
\begin{split}
\mathbf{B}(\mathbf{b}) = (B_1(\mathbf{b}_1), \dots, B_d(\mathbf{b}_d))^\top, \\
B_i(\mathbf{b}_i) 
\;=\; 
\sum_{p=0}^{m-1} 2^{m-1-p} \, b_{i,p}, \\
\mathbf{b}_i \;=\; \bigl(b_{i,0}, b_{i,1},\dots, b_{i,m-1}\bigr)^\top.
\end{split}
\end{equation}

\subsubsection{Quadratic Cost Function}

The sum of squared errors for the discretized weights is 
\begin{equation}
S^{(1)}(\mathbf{b}) 
\;=\; 
\sum_{n=1}^N \Bigl(y_n - \mathbf{w^{(1)}}(\mathbf{b})^\top \mathbf{x}_n\Bigr)^2.
\end{equation}
Expanding this, we get
\begin{equation}
S^{(1)}(\mathbf{b})
\;=\;
\sum_{n=1}^N \Bigl(y_n - \boldsymbol{\ell^{(0)}}^\top \mathbf{x}_n 
- \mathbf{B}(\mathbf{b})^\top \mathbf{D^{(0)}}^\top \mathbf{x}_n \Bigr)^2.
\end{equation}
Since each component of \(\mathbf{B}(\mathbf{b})\) is linear in the binary variables, \(S^{(1)}(\mathbf{b})\) becomes a polynomial (up to quadratic order) in those bits:
\begin{equation}
S^{(1)}(\mathbf{b})
\;=\;
\alpha^{(1)}
\;+\;
\sum_{r} \gamma_r^{(1)} \, b_r
\;+\;
\sum_{r < s} \Gamma_{r,s}^{(1)} \, b_r\,b_s,
\end{equation}
where the summations over \(r, s\) run over all binary variables in \(\mathbf{b}\), and \(\alpha^{(1)}, \gamma_r^{(1)}, \Gamma_{r,s}^{(1)}\) are real coefficients derived from \(\mathbf{X}, \mathbf{y}, \boldsymbol{\ell^{(0)}}, \mathbf{D^{(0)}}\). 

This polynomial is in fact strictly quadratic, because each \(b_r^2 = b_r\) as \(b_r \in \{0,1\}\). Thus we have effectively mapped the least squares cost onto a QUBO:
\begin{equation} \label{eqn:qubo}
\mathbf{b^{(1)*}}
\;=\;
\min_{\mathbf{b} \in \{0,1\}^{d \times m}} 
\;\; \alpha^{(1)} 
\;+\;
\sum_{r=1}^{d m} \gamma_r^{(1)} \, b_r
\;+\;
\sum_{1 \,\le\, r < s \,\le\, dm} 
\Gamma_{r,s}^{(1)} \, b_r\,b_s.
\end{equation}

Given Eq.~\ref{eqn:qubo}, a quantum annealer can directly solve QUBO problems by mapping them onto its native hardware architecture, where the cost function is minimized adiabatically through quantum tunneling. The annealer attempts to find a low-energy configuration of the binary variables that corresponds to the optimal solution of the given problem. However, the total number of binary variables that can be encoded is constrained by the available qubits, making it challenging to achieve high precision (i.e., large \(m\)) in naive discretization, as each additional bit per weight significantly increases qubit requirements.

\subsection{\algoName Algorithm}

To address the qubit-precision trade-off, we propose an iterative zoom-in approach that uses a small number of bits \(m \ge 1\) per weight but refines the solution region over multiple iterations. 
Let \(\ell_i^{(0)}\) and \(u_i^{(0)}\) be the initial bounds for weight \(w_i\). At iteration \(k=1,\dots,K\), each weight \(w_i\) is represented in \([\ell_i^{(k-1)},\, u_i^{(k-1)}]\) with \(m\) bits. With $\Delta_i^{(k-1)} = u_i^{(k-1)} - \ell_i^{(k-1)}$, the discretization step size is evaluated as
\[
   \delta_i^{(k-1)} 
   \;=\; 
   \frac{\Delta_i^{(k-1)}}{2^m-1}.
\]
Thus,
\[
   \mathbf{w^{(k)}}(\mathbf{b}) 
   \;=\; 
   \boldsymbol{\ell^{(k-1)}} 
   \;+\; \mathbf{D^{(k-1)}} \mathbf{B}(\mathbf{b})
\]
Construct the QUBO from 
\[
   S^{(k)}(\mathbf{b}) 
   \;=\; 
   \sum_{i=1}^N 
   \bigl( y_i - \mathbf{w^{(k)}}(\mathbf{b})^\top \mathbf{x}_i\bigr)^2
\]
with the updated bounds. Simplify the expression so that powers of binary variables are reduced via \(b_{j,r}^2 = b_{j,r}\).
Solve the QUBO using a quantum annealer.
   Retrieve the optimal binary solution \(\mathbf{b}^{(k)*}\). Compute 
   \[
   w_i^{(k)*} 
   \;=\;
   w_i(\mathbf{b}_i^{(k)*}), 
   \quad i=1,\dots,d.
   \]
For each \(i\):
\begin{equation} \label{eqn:search_interval_update}
\begin{split}
   \ell_i^{(k)} 
   \;=\;
   w_i^{(k)*} 
   \;-\;
   (\frac{\delta_i^{(k-1)}}{2{f(m)}}).
   \quad
   u_i^{(k)} 
   \;=\;
   w_i^{(k)*} 
   \;+\;
   (\frac{\delta_i^{(k-1)}}{2{f(m)}}), \\
   \text{where} \quad f(x) =
                    \begin{cases} 
                    2, & \text{if } x = 1, \\
                    1, & \text{otherwise}.
                    \end{cases}
\end{split}
\end{equation}
This ensures the next iteration’s search space is an interval of width \(\Delta_i^{(k)}\) centered at the best estimate \(w_i^{(k)*}\).
Observe that, for $m=1$, $\delta_i^{k} = \Delta_i^{k} \implies w_i^{(k)} \in \{\ell_i^{(k-1)}, u_i^{(k-1)}\}$, thus the step function $f(x)$ helps in shrinking the search space exponentially.
Repeat until the maximum number of iterations \(K\) is reached, or until a convergence criterion is met (e.g., changes in the loss below a threshold).

This procedure scales linearly with \(d\) in qubit usage (only \(d \times m\) qubits are needed), yet allows an effective exponential zoom-in on the candidate solution region over multiple iterations, thereby improving precision without globally increasing the qubit count. The pseudocode of the algorithm is given in \ref{algo:iterative_QA_LS}.

\begin{algorithm}[ht]
    \caption{i-QLS: Iterative Quantum-Assisted Least Squares}
    \label{algo:iterative_QA_LS}
    \KwIn{$\mathbf{X} \in \mathbb{R}^{N\times d}$ (feature matrix), $\mathbf{y} \in \mathbb{R}^{N}$ (target values), 
          bits per weight $m$, max iterations $K$, initial bounds $\ell_i^{(0)}, u_i^{(0)}$ for each $i=1,\dots,d$.}
    
    \For{$k = 1$ \textbf{to} $K$}{
        Compute step size: $\delta_i^{(k-1)} \gets \frac{u_i^{(k-1)} - \ell_i^{(k-1)}}{2^m-1},\; i=1,\dots,d.$ \\
        
        Construct weight representation $w_i^{(k)}(\mathbf{b})$ (Eq. \ref{eqn:weight_vector_discretization}) \\
        
        Formulate QUBO cost function $S^{(k)}(\mathbf{b})$. \\
        
        Solve $\min S^{(k)}(\mathbf{b})$ using a quantum annealer. \\
        
        Extract optimal solution: $\mathbf{b}^{(k)*} \gets \arg \min_{\mathbf{b}} S^{(k)}(\mathbf{b})$. \\
        
        Compute updated weight estimates: $w_i^* \gets w_i(\mathbf{b}^{(k)*})$. \\
        
        Update bounds: \\
        \Indp
        $\ell_i^{(k)} \gets w_i^{(k)*} - \frac{1}{2{f(m)}} \delta_i^{(k)},\quad u_i^{(k)} \gets w_i^{(k)*} + \frac{1}{2{f(m)}} \delta_i^{(k)}$.
        \Indm
    }
    
    \KwOut{$\mathbf{w}^{(K)*} = (w_1^{(K)*}, \dots, w_d^{(K)*})$ (optimal weights).}
\end{algorithm}

\begin{lemma}\label{lemma:convergence}
If the underlying least squares problem is well-posed (i.e., there exists a unique optimal weight vector $w^*$ such that $y=Xw^*$) and $w_i^* \in [l_i^{(0)}, u_i^{(0)}] \forall i$, then the mean squared error evaluated from the weights estimated at each iteration exponentially converges to $0$ as $k \to \infty$.
\end{lemma}

\begin{proof}

At iteration $k=0$, the search interval for weight $w_i$ is given by $[l_i^{(0)}, u_i^{(0)}]$ with width $\Delta_i^{(0)}$. In the first iteration, the algorithm discretizes this interval into $2^m$ equally spaced values. Let the discretization step size be 
\[
\delta_i^{(0)} = \frac{\Delta_i^{(0)}}{2^m-1}.
\]
The QUBO formulation selects the discrete value closest to the true optimum, denoted by $w_i^{(1)}$. The algorithm then updates the search interval to be centered at $w_i^{(1)}$ with new bounds
\[
l_i^{(1)} = w_i^{(1)} - \frac{\delta_i^{(0)}}{2{f(m)}}, \quad u_i^{(1)} = w_i^{(1)} + \frac{\delta_i^{(0)}}{2{f(m)}},
\]
so that the new interval width is 
\[
\Delta_i^{(1)} = \frac{\delta_i^{(0)}}{f(m)} = \frac{\Delta_i^{(0)}}{{f(m)}(2^m-1)}.
\]

By repeating this process iteratively, the width of the interval after $k$ iterations is
\begin{equation} \label{eqn:search_interval_shrinking}
\Delta_i^{(k)} = \frac{\Delta_i^{(k-1)}}{{f(m)}(2^m-1)} = \frac{\Delta_i^{(0)}}{{({f(m)}(2^m-1))}^k}.
\end{equation}
Since ${f(m)(2^m-1)}> 1$ for any $m \ge 1$, it follows that
\[
\lim_{k \to \infty} \Delta_i^{(k)} = 0.
\]
Because the search interval shrinks to a single point, the weight estimate $w_i^{(k)}$ converges to a unique value, which must coincide with the optimal weight $w_i^*$ provided the model is correctly specified. 

Furthermore, since the mean squared error is a continuous function of the weights, it follows that
\[
\lim_{k\to\infty} E^{(k)} = \lim_{k\to\infty} \|Xw^{(k)} - y\|^2 = \|Xw^* - y\|^2.
\]
Under the assumption that the least squares problem is consistent (or that $w^*$ minimizes the error), we have $\|Xw^* - y\|^2 = 0$, which implies
\[
\lim_{k\to\infty} E^{(k)} = 0.
\]
\end{proof}

\section{Experiments}

We present experiments evaluating convergence behavior, scalability, and applicability of \algoName to non-linear regression.

\subsection{Convergence and Search Space Refinement}

Figure~\ref{fig:two_feature_data} shows the synthetic dataset generated for a two-feature linear regression problem. Using \algoName, we observe exponential decay of the mean squared error (MSE) over iterations (Figure~\ref{fig:mse_convergence}), with faster convergence for higher bit precision. For example, 6-bit precision achieves near-zero MSE at $k=9$ iterations, while 2-bit precision plateaus around $10^{-9}$. This supports our theoretical convergence rate.

Figure~\ref{fig:search_space_shrinkage} visualizes the search space refinement across iterations. The weight bounds for $w_1$ and $w_2$ contract exponentially, with higher bit precision enabling faster narrowing of the feasible region.

\begin{figure}[ht]
    \centering
    \begin{subfigure}{0.4\textwidth}
        \centering
        \includegraphics[width=\textwidth]{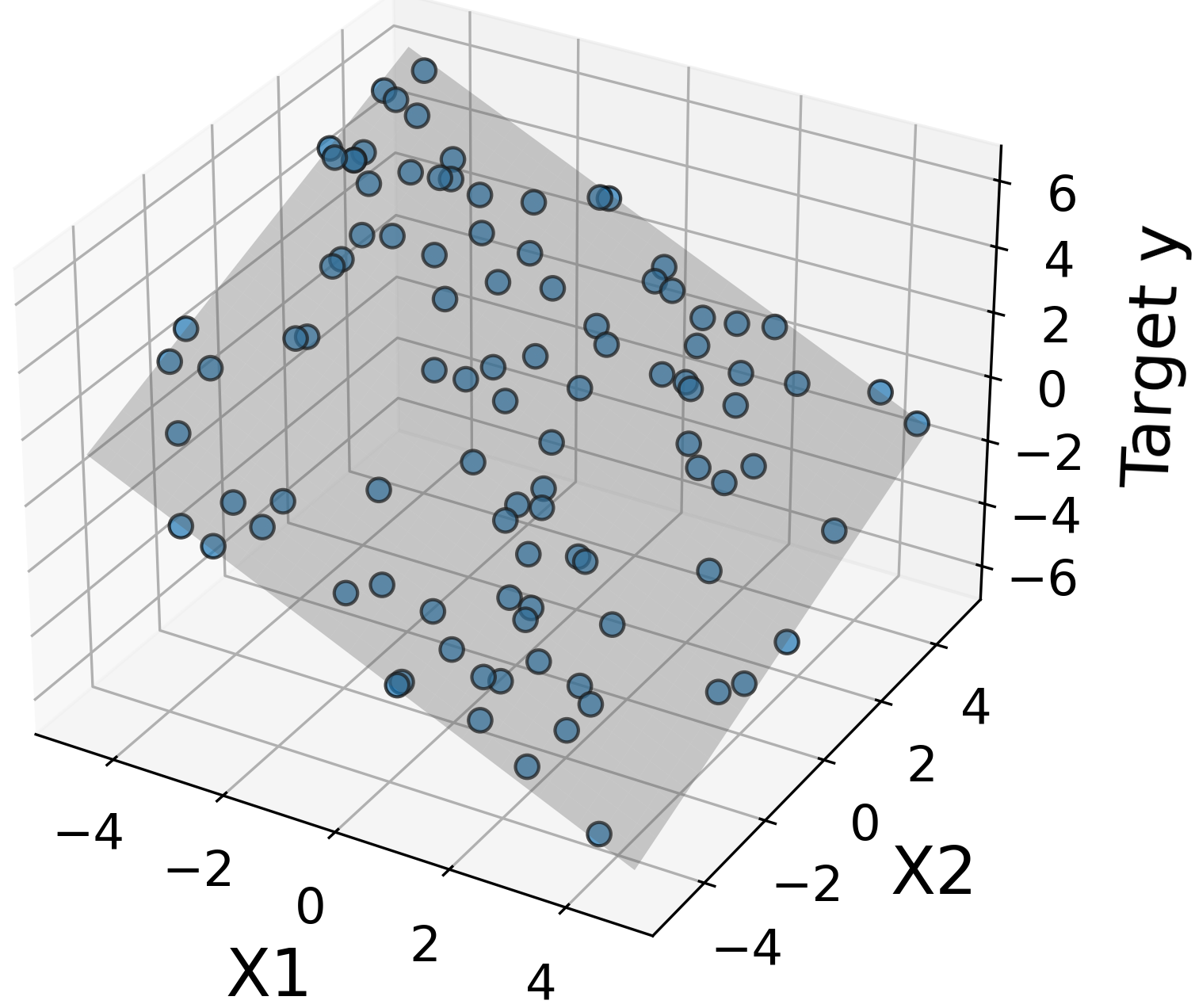}
        \caption{
        \textbf{Two-Feature Linear Data.} 
        100 data points $(x_i, y_i)$ where $x_i \in [-5,5]^2$ sampled from a linear model.}
        \label{fig:two_feature_data}
    \end{subfigure}
    \hspace{1pt}
    \begin{subfigure}{0.56\textwidth}
        \centering
        \includegraphics[width=\textwidth]{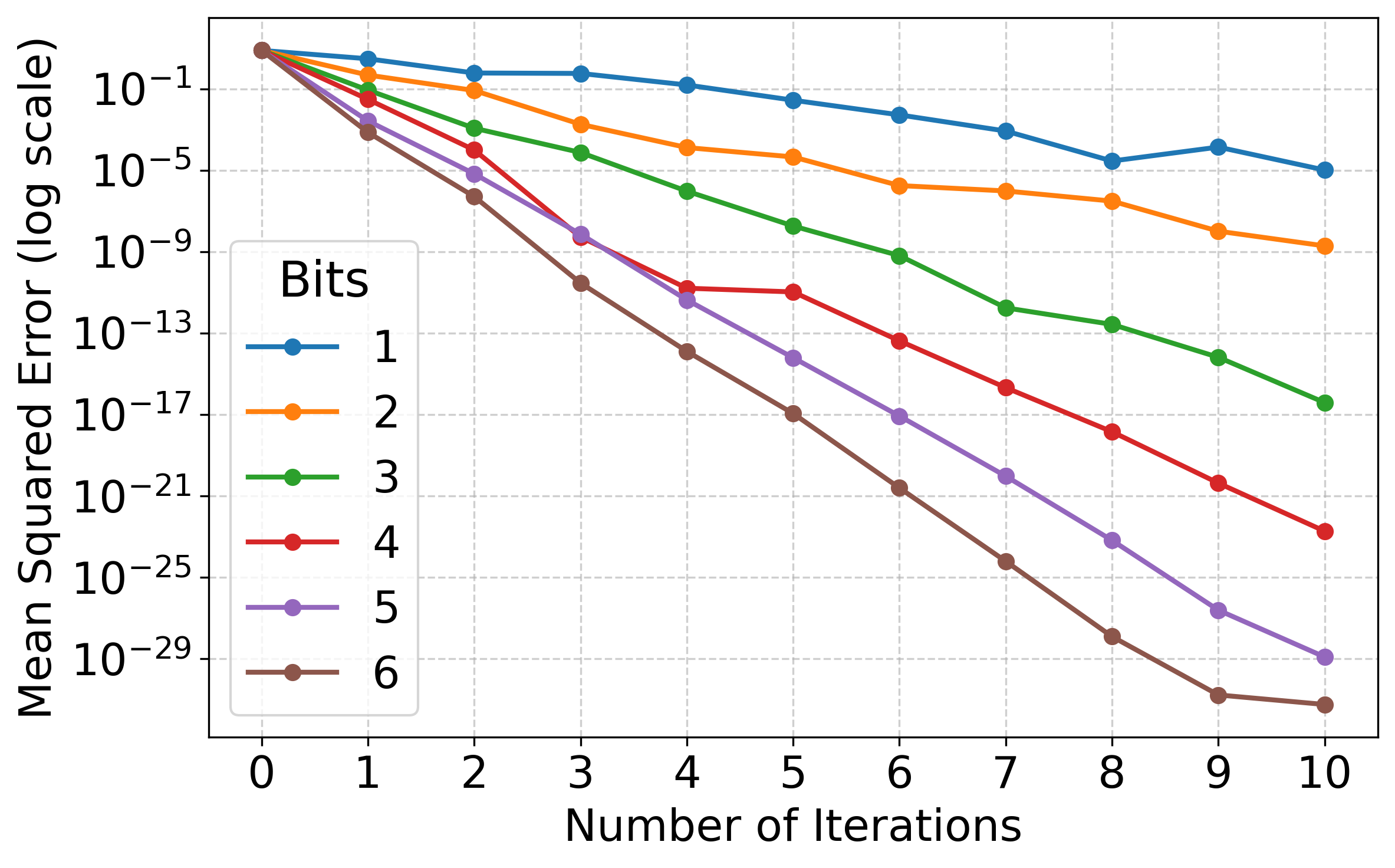}
        \caption{\textbf{MSE Convergence Over Iterations.} 
        Exponential decay of MSE with iterations, improving with increasing bits per weight.}
        \label{fig:mse_convergence}
    \end{subfigure}
    
    \medskip
    
    \begin{subfigure}{\textwidth}
        \centering
        \includegraphics[width=\textwidth]{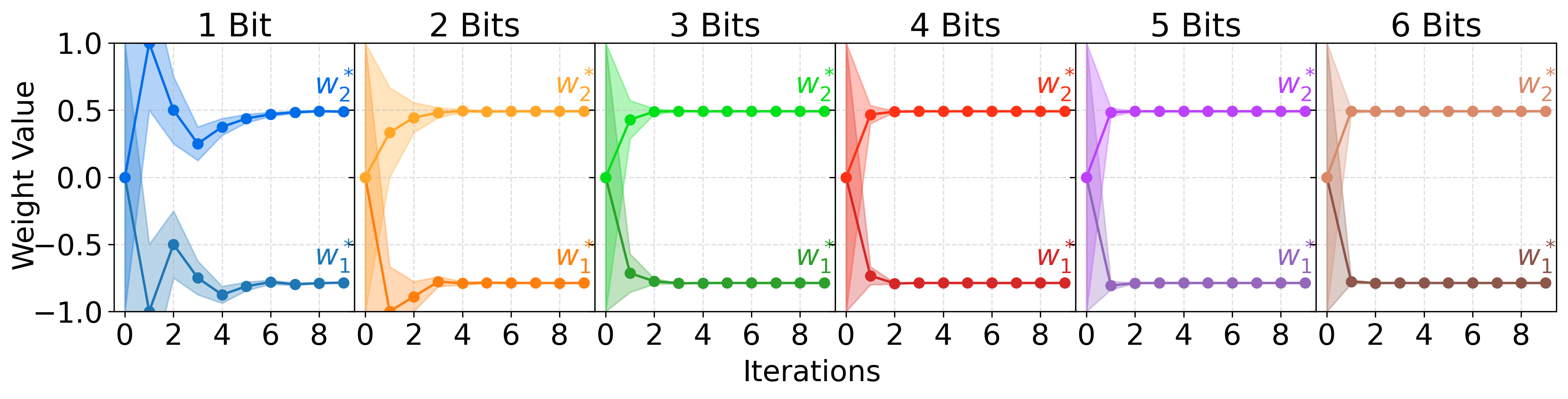}
        \caption{\textbf{Exponential Shrinking of the Search Space.} 
        The weight search space for $w_1$ and $w_2$ narrows exponentially with iterations.
        Higher bit precision enables finer updates, accelerating convergence.}
        \label{fig:search_space_shrinkage}
    \end{subfigure}    
    \caption{%
    Assessment of convergence rate and accuracy of \algoName} 
    \vspace{-12pt}
\end{figure}

\subsection{Non-Linear Function Approximation}

We applied \algoName to approximate five standard non-linear functions using linear splines with 20 knots. Figure~\ref{fig:spline_nonlinear} shows the results. The blue scatter points denote ground-truth values; the green curves illustrate the iterative spline fits. Intermediate fits (light green) gradually approach the final estimate (dark green) over 10 iterations. The piecewise linear nature of the spline basis leads to a segmented, angular approximation pattern.

While higher-order splines could achieve smoother fits, they would require more parameters and qubits. Nonetheless, these results demonstrate the ability of \algoName to extend beyond linear models and capture non-linear patterns.

\begin{figure}[ht]
    \centering
    \includegraphics[width=\textwidth]{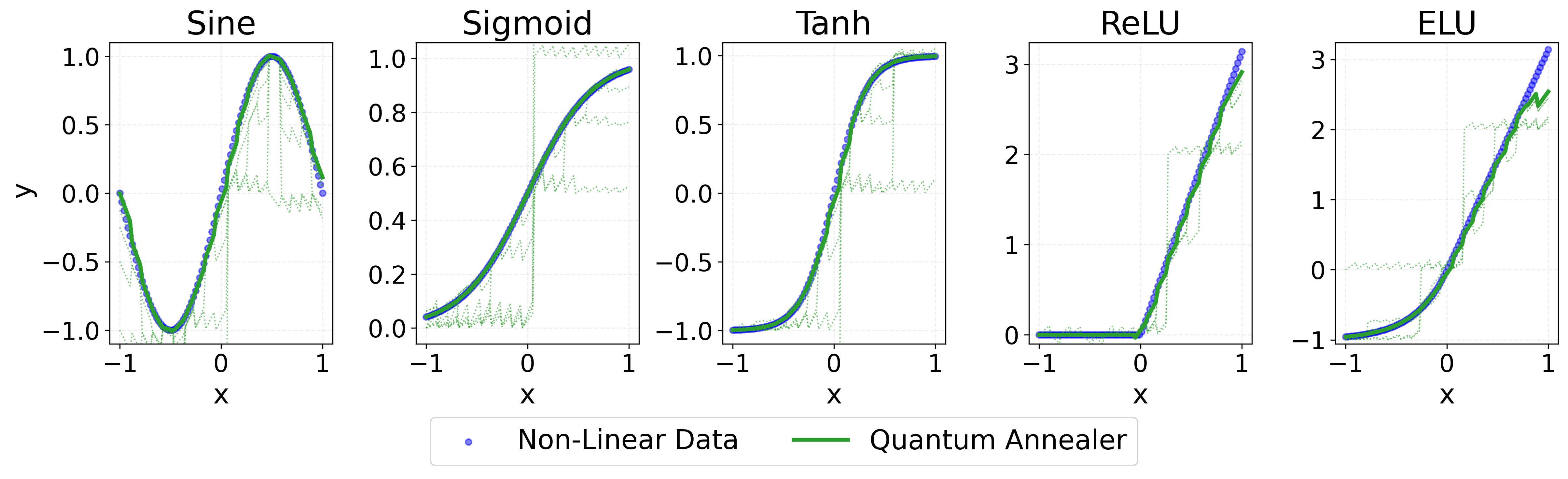}
    \caption{%
    \textbf{Spline-Based Approximation of Non-Linear Functions.} 
    The green curves illustrate the iterative refinement of the regression fit obtained using our quantum-assisted least squares approach with \textit{linear splines} (20 knots), one qubit per parameter, and up to 10 iterations. The lighter green lines correspond to intermediate fits across iterations, while the darker green curve represents the final approximation after 10 iterations.}
    \label{fig:spline_nonlinear}
\end{figure}

\section{Discussion and Conclusion}

In this work, we introduced i-QLS, an iterative quantum-assisted least squares algorithm that improves scalability over prior single-shot QUBO methods. By iteratively refining the search space around the best solution, the interval contracts by a factor of ${f(m)(2^m-1)}$ per iteration, achieving convergence within
\[
K = \mathcal{O}\left( \log_{{f(m)(2^m-1)}} \frac{1}{\epsilon} \right)
\]
iterations for precision $\epsilon$. Higher bit precision reduces $K$, but increases quantum resource demands, as reflected in Figure~\ref{fig:mse_convergence}, showing a trade-off between convergence speed and hardware feasibility.

Despite theoretical guarantees of convergence (even with $m=1$ as $k \to \infty$), practical limitations may affect performance. In noisy hardware or imperfect solvers, an iteration may select a suboptimal weight, shrinking the subsequent search space to exclude the true optimum. This risk increases at low bit precision, as seen in Figure~\ref{fig:mse_convergence}: for 1-bit precision, MSE temporarily increases at iteration 9 before declining at iteration 10. Such fluctuations emphasize the need for sufficient iterations to ensure convergence.

While i-QLS demonstrates strong scalability—successfully solving problems up to 175 features on a D-Wave Advantage annealer—its runtime is currently dominated by remote access latency and queue delays \cite{10669751}. Local hardware access could mitigate this overhead, making i-QLS more competitive for time-sensitive applications.

Beyond linear regression, we showed that i-QLS extends to spline-based non-linear regression, achieving accurate fits despite using only linear splines and low qubit counts. Higher-order splines remain a promising extension, though requiring additional qubits.

In summary, i-QLS offers a scalable and adaptable quantum-assisted regression framework. Future work will investigate adaptive bit allocation and iteration strategies to balance precision, convergence, and quantum resource constraints. The use of qubit-efficient techniques for solving QUBO problems on gate-based quantum computers opens new avenues for advancing quantum-assisted optimization methods \cite{10821431}.

\begin{credits}
\subsubsection{Code Availability.} The code associated with this paper is available at: \\
\href{https://github.com/supreethmv/i-QLS}{https://github.com/supreethmv/i-QLS}
\end{credits}

\begin{credits}
\subsubsection{\ackname} This work has been partially funded by the German Ministry for Education and Research (BMB+F) in the project QAIAC-QAI2C under grant 13N17167.

\vspace{-12pt}
\end{credits}

\bibliographystyle{splncs04}
\bibliography{references}

\end{document}